\documentclass[aps,pra,showpacs,twocolumn,superscriptaddress]{revtex4-1}
\usepackage{amscd,amsfonts,color,bbm, amssymb, amsthm}
\usepackage{amsmath,amscd,amsfonts,amssymb}
\newtheorem{prop}{Proposition}[section]

\newcommand{\bra}[1]{\left\langle{#1}\right\vert}
\newcommand{\ket}[1]{\left\vert{#1}\right\rangle}
\usepackage{graphicx,amsfonts,amsmath,bm,color}
\usepackage[pdftex,bookmarks,colorlinks]{hyperref}
\definecolor{darkblue}{rgb}{0.0,0.0,0.3}
\hypersetup{colorlinks=true,citecolor=darkblue,
unicode=true,pdftitle=crow,pdfauthor=Arvind,bookmarks=false}
\begin{document}
\title{Measurement-based local quantum filters and
their ability to transform quantum entanglement}
\author{Debmalya Das}
\email{debmalya@iisermohali.ac.in}
\affiliation{Department of Physical Sciences, 
Indian Institute of Science Education \& Research
(IISER), Mohali, India}
\author{Ritabrata Sengupta}
\email{rb@isid.ac.in}
\affiliation{Department of Mathematical Sciences, 
Indian Institute of Science Education \& Research
(IISER), Mohali, India}
\affiliation{Theoretical Statistics and Mathematics Unit, Indian  
Statistical Institute, Delhi Centre, New Delhi, India}
\author{Arvind}
\email{arvind@iisermohali.ac.in}
\affiliation{Department of Physical Sciences, 
Indian Institute of Science Education \&
Research (IISER), Mohali, India}
\begin{abstract}
We introduce local filters as a means to detect
the entanglement of bound entangled states which
do not yield to detection by witnesses based on
positive (P) maps which are not completely
positive (CP).  We demonstrate how such
non-detectable bound entangled states can be
locally filtered into detectable bound entangled
states.  Specifically, we show that a bound
entangled state in the orthogonal complement of
the unextendible product bases (UPB), can be
locally filtered into a another bound entangled
state that is detectable by the Choi map.  We
reinterpret these filters as local measurements on
locally extended Hilbert spaces. We give explicit
constructions of a measurement-based implementation
of these filters  for  $2 \otimes 2$ and  $3
\otimes 3$ systems. This  provides us with a
physical mechanism to implement such local
filters.
\end{abstract} 
\maketitle
\section{Introduction}
\label{intro}
Ever since its introduction by
Schr\"{o}dinger~\cite{shro1,shro2} in the context
of the EPR paradox~\cite{epr}, quantum
entanglement has played a central role in quantum
theory.  While entanglement is responsible for the
non-classical correlations leading to the
violation of Bell's
inequalities~\cite{Bell1964a,Bell}, it also plays
a key role in quantum computing where it is
connected with the  exponential advantage of
quantum algorithms over their classical
counterparts~\cite{NC}.

Studies of entanglement  have led to a well
developed mathematical theory of entanglement
where  positive maps (P) which are not completely
positive (CP)~\cite{cho1, kye, kye1, kossak1,
kossak4, 1108.0130, PhysRevA.84.024302,
MR3280004}  and unextendable product bases
(UPB)~\cite{B2, T2, MR2907636} play an important
role~\cite{T1, PhysRevA.84.032328,
PhysRevA.87.012318, PhysRevA.87.064302}.  These
mathematical advances have led to the discovery of
bound entangled states~\cite{boundent}: states
from which one cannot distill EPR pairs although
they are still provably non-separable.

Quantum states (pure or mixed) are represented  by
positive definite Hermitian operators $\rho \in
\mathcal{B(H)}$ with unit trace. For the special
case when the rank of $\rho$ is one, it
represents a pure state.  For a bipartite
composite system where states  are defined on
$\mathcal{B}(\mathcal{H}_A \otimes
\mathcal{H}_B)$, a state $\rho$  is said to be
separable if it can be written as a convex sum:
\begin{equation}
\rho=\sum_{j} p_j \rho_j^A \otimes
\rho_j^B, \quad p_j>0, \quad \sum_j p_j=1; 
\end{equation}
where $\rho_j^A$ and $\rho_j^B$ are states in
$\mathcal{B}(\mathcal{H}_A)$ and
$\mathcal{B}(\mathcal{H}_B)$ respectively. A state
is entangled, if it cannot be written in the above
form.  The fundamental problem of determining
whether a given state $\rho$  is separable or
entangled remains open for general states of
systems which are beyond $2\otimes2$ and $2\otimes 3$.

Allowed quantum evolutions are  those P maps which
are CP. P maps which are not CP are not physically
allowed quantum evolutions,  because entangled
states can lose their positivity when  such a map
is applied to one part of the system.  Therefore
such maps  act as entanglement witnesses.  The
partial transpose operation which is a particular
entanglement witness, plays an important role in
the identification of entangled states~\cite{P1}.
States which reveal their entanglement by
acquiring one or more negative eigenvalues under
partial transposition are called NPT (not-positive
under partial transpose) while the rest are called
PPT (positive under partial transpose).  NPT
states are entangled while PPT states can be
either entangled or separable.  In $2\otimes2$ and
$2\otimes 3$  dimensional systems, it has been
shown that a state is separable if and only if it
is PPT~\cite{P1,H1}. Therefore, in dimensions $3
\otimes 3$ and larger, there are entangled states
which are PPT~\cite{guth}.  These states in
general require other entanglement witnesses to
implicate their entanglement and cannot be
distilled to give EPR pairs. Such states (with
non-distillable entanglement) are called PPT or
bound entangled states.

Choi map and its generalizations have been used to
detect the entanglement of certain classes of
bound entangled states~\cite{simon}. However, even
for the $3\otimes3$ system, the detection of bound
entangled states is far from complete. Local
operations, including measurement and filtering,
cannot alter the status of the state from PPT to
NPT and therefore can be used to convert  one NPT
state to another NPT state or one PPT state to
another PPT state.  Gisin et. al. showed that,
starting with a mixed entangled state of a
$2\otimes 2$ system that does not violate Bell
inequalities, one can set up a local filtration
scheme based on measurements such that the
filtered states violate Bell
inequalities~\cite{MR1371104}. In this case only
NPT states were involved as there are no PPT
entangled states for the $2\otimes 2$ case.  These
results have been extended, used and
experimentally validated by a number  of
researchers~\cite{PhysRevLett.111.160402,PhysRevA.64.010101,
PhysRevLett.89.170401,1751,PhysRevLett.111.160402,Kwiat}.
Our work is a generalization and extension of
these results into the domain of PPT states of
$3\otimes 3$ systems.  In this work, we introduce
local filters which convert PPT entangled states
not detectable by the Choi map, to states which
are detectable by the Choi map. In particular we
are able to show that the PPT states obtained from
the UPB can be converted into states detectable by
the Choi map.  Furthermore, we provide an explicit
scheme for implementation of our filtration
protocol via local projective measurements
involving local ancillas.  Although the notation
of entanglement is well defined for infinite
dimensional spaces, we restrict ourselves to
finite dimensional Hilbert spaces in this paper.

The material in this paper is arranged as follows:
In Section~\ref{filtration} we describe our local
filtration scheme. Two examples are taken up in
the Section~\ref{filtration-ent} where such
schemes are used to manipulate entangled states.
Section~\ref{implementation} describes a
measurement-based scheme to realize the local
filters while Section~\ref{implementation-general}
describes the general scheme.
Section~\ref{implementation-two} describes the
implementation of filters used by Gisin and
Section~\ref{implementation-three} describes the
implementation of filters on three-level systems.
Section~\ref{conc} offers some concluding remarks.
\section{Local filtration and entanglement
detection}
\label{filtration}
\subsection{Local filters}
\label{filtration-local}
Local filters are local non-unitary operators
represented by  $L\otimes M$ where $L$ and $M$ are
invertible operators acting in the state spaces of
their respective systems.  Given a bipartite
quantum state $\rho\in\mathcal{B}(\mathcal{H}_A \otimes
\mathcal{H}_B)$, the filter acts
on the state giving a new state 
\begin{equation}
\rho^{f} = (L\otimes M)\rho (L\otimes M)^\dag
\label{filter}
\end{equation} 
which is a positive Hermitian operator belonging
to the same space and its trace can be brought to
one by dividing by an appropriate positive number.
For every invertible set of operators $L$ and $M$
we thus have a filter.
\begin{prop}
Let $L$ and $M$ be two full rank operators. Then the
map $\rho \mapsto (L\otimes M) \rho  (L\otimes
M)^\dag$ does not change the Schmidt number of the state. 
\end{prop}
\begin{proof}
Terhal and Horodecki~\cite{T4} defined
the Schmidt rank of a general density matrix.
A bipartite state $\rho$ has Schmidt rank $k$ if 
\begin{enumerate}
\item For any ensemble decomposition of $\rho$ as $\{p_j\geq
0,|\psi_j\rangle\}$ 
where $\rho=\sum_j p_j |\psi_j\rangle \langle \psi_j|$ at
least one of the vectors $|\psi_j\rangle$ has  Schmidt rank
$k$. 
\item There exists a decomposition of $\rho$ where all
vectors $\{|\psi_j\rangle\}$ in the decomposition have
a Schmidt rank at most $k$. 
\end{enumerate}
Therefore, we need only to show that the Schmidt
number of pure bipartite states 
$\ket{\psi}=\sum_j \lambda_j \ket{e_j}\otimes
\ket{f_j}$ is 
invariant under the operation
$\ket{\psi} \mapsto (L \otimes M) \ket{\psi}$.

The Schmidt
rank (SR) of $\ket{\psi}$ is the matrix rank of $\sum_j \lambda_j
\ket{f_j}\bra{e_j}$. Thus 
\begin{equation}
{\rm SR} \left( ( L \otimes M ) \ket{\psi}\right) =
{\rm Rank} \sum_j \lambda_j M \ket{f_j}\bra{e_j} L^\dag.
\end{equation}
Let $L=U_1D_1V_1$ and $M=U_2D_2V_2$ be the 
singular value decompositions of $L$ and $M$
respectively, where $U_1,V_1,U_2,V_2$
are unitary operators. 
Then 
\begin{eqnarray}
{\rm SR} \left( (L \otimes M )\ket{\psi}\right) &=&
{\rm Rank} \sum_j \lambda_j U_2D_2V_2 \ket{f_j}\bra{e_j} V_1^\dag
D_1^\dag U_1^\dag\nonumber \\
&=& {\rm Rank} \sum_j \lambda_j D_2\ket{f_j^{\prime}}\bra{e_j^{\prime}}
D_1^\dag
\end{eqnarray}
where $\ket{e_j^{\prime}}=V_1\ket{e_j}$ and
$\ket{f_j^{\prime}}=V_2\ket{f_j}$ are mutually orthogonal bases of
the first and second systems respectively. Since $L$
and $M$
are of full rank, the diagonal matrices $D_1$ and $D_2$ are
also of full rank, and the above assertion holds,
i.e. ${\rm SR}\ket{\psi}$ is
invariant under these operations. 
\end{proof}
The above proposition further shows that,
entanglement is not created or destroyed by the  above
operations.

Let us choose a standard basis $\{\vert j \rangle
\}$, $j=1\cdots n$ in an $n$-dimensional state space
$\mathcal{H}$. Any density operator $\rho \in 
\mathcal{B(H)}$ can then be  written as
$\rho=\displaystyle \sum_{i,j}\rho_{ij}\vert i \rangle \langle
j\vert$. Transpose operation is defined through
its action on $\rho$
\begin{equation}
\rho \stackrel{T}{\longrightarrow}\rho^T=\sum_{ij}\rho_{ji}
\vert i \rangle \langle j \vert
\label{transpose}
\end{equation}
A bipartite state $\rho$ is defined to be 
PPT if and only if $(\bm{1}\otimes T)\rho\geq0$ where $T$ is
the transpose operation defined on $\mathcal{H}_B$
as described in Equation~(\ref{transpose}).  
\begin{prop}
The  PPT or NPT character of a 
state is invariant
under an invertible local filtration operation.
\end{prop}
\begin{proof}
Now consider  $\{|j\rangle\}$, $j=1\cdots n_1$ to be 
be the standard basis in
$\mathcal{H}_A$. We can write $\rho =
\sum_{i,j} |i\rangle\langle j|\otimes \rho_{ij}$,
where $\rho_{ij}\in\mathcal{B(\mathcal{H}_B)}$. We
then have \begin{equation} \rho^f=(L\otimes M)\rho
(L\otimes M)^\dag = \sum_{i,j} L|i\rangle\langle
j|L^\dag \otimes M\rho_{ij}M^\dag \end{equation}
After the application of the transpose operation
on the second system we have 
\begin{eqnarray}
(\bm{1}\otimes T) \rho^f &=& \sum_{i,j}
L|i\rangle\langle j|L^\dag \otimes
(M\rho_{ij}M^\dag)^T\nonumber \nonumber \\ &=&
\sum_{i,j} L|i\rangle\langle j|L^\dag \otimes
\bar{M}\rho_{i,j}^T\bar{M}^\dag \nonumber \\ &=&
(L\otimes\bar{M}) \left(\sum_{i,j} |i\rangle\langle
j|\otimes \rho_{ij}^T\right) (L\otimes\bar{M})^\dag
\nonumber \\ &=&(L\otimes\bar{M}) ((\bm{1}\otimes
T)\rho) (L\otimes\bar{M})^\dag \geq 0.
\end{eqnarray} 
where $\bar{M}$ is the complex
conjugate of $M$.  This proves that PPT states
remain PPT under a local filter defined in
Equation~(\ref{filter}). A similar argument holds
for NPT states.
\end{proof}
If the original state is entangled, the nature of
its entanglement, (NPT or PPT) does not change under
the filtering operation. Therefore for a given PPT
entangled state, the filtered state is another PPT
entangled state. It may turn out that even if the
entanglement of the original state is not detectable by a
given entanglement witness, the filtered state
reveals its entanglement by the same witness. 
This thus allows us to generate new PPT entangled
states from the old ones.
\subsection{Entanglement witnesses and local filters}
\label{filtration-ent}
Two maps due to Choi~\cite{choi2} 	
(which are P maps that are not
CP) are defined through their action on a
$3\times3$ matrix as:
\begin{eqnarray}
\phi:[[a_{ij}]] &\mapsto& \frac{1}{2}
\begin{pmatrix}
a_{11}+a_{33} & -a_{12} & -a_{13} \\
-a_{21} & a_{22}+ a_{11} & -a_{23} \\
-a_{31} & -a_{32} & a_{33}+ a_{22}
\end{pmatrix}  \label{eqn:1}\\
\psi:[[a_{ij}]] &\mapsto& \frac{1}{2}
\begin{pmatrix}
a_{11}+a_{22} & -a_{12} & -a_{13} \\
-a_{21} & a_{22}+ a_{33} & -a_{23} \\
-a_{31} & -a_{32} & a_{33}+ a_{11}
\end{pmatrix}. \label{eqn:2}
\end{eqnarray}
These maps  provide us with  important
entanglement witnesses for PPT entangled states. 
Extensions of these maps have been used to unearth
new PPT entangled states and to detect
entanglement of PPT states formed out of UPB~\cite{PhysRevA.84.032328,
PhysRevA.87.012318}.  We recast these results in
terms of local filters and show that the PPT
entangled states described in~\cite{PhysRevA.84.032328}
can be locally filtered into a state that is
detectable by the Choi map. Similarly we  show
that the PPT entangled states in the orthogonal
complement of UPB can be locally filtered into
states detectable by the Choi map.

Consider a density operator for the system
$\mathcal{B}(\mathbb{C}^3 \otimes \mathbb{C}^3)$,
defined by two parameters $t$ and $x$.
\begin{equation}\label{choiexample}
\rho(x,t)= K
\left(
\begin{array}{ccc|ccc|ccc}
 1+t & 0 & 0 & 0 & x & 0 & 0 & 0 & x \\
 0 & t & 0 & x & 0 & 0 & 0 & 0 & 0 \\
 0 & 0 & \frac{1}{t} & 0 & 0 & 0 & x & 0 & 0 \\\hline
 0 & x & 0 & \frac{1}{t} & 0 & 0 & 0 & 0 & 0 \\
 x & 0 & 0 & 0 & 1+t & 0 & 0 & 0 & x \\
 0 & 0 & 0 & 0 & 0 & t & 0 & x & 0 \\\hline
 0 & 0 & x & 0 & 0 & 0 & 1 & 0 & 0 \\
 0 & 0 & 0 & 0 & 0 & x & 0 & \frac{1}{t} & 0 \\
 x & 0 & 0 & 0 & x & 0 & 0 & 0 & 1
\end{array}
\right).
\end{equation}
where $K=\frac{1}{4+\frac{3}{t}+4 t}$ is the
normalization constant and this $\rho$ is a unit
trace density operator for $t>0$ and $0\leq
x\leq1$. The state is entangled for a range of
values of $x$ and $t$; however it is not always
detected by the Choi maps given in
equations~(\ref{eqn:1}) and~(\ref{eqn:2}).  For
instance, set $t =\frac{1}{20}$, then for $0.6044
< x < 0.6554$, the  state is not detectable by the
Choi map.
Consider a local filter 
\begin{equation}
L_3\otimes M_3 = \left( \begin{array}{ccc}
1&0&0\\0&\frac{5}{8}&0\\ 0&0&
\frac{5}{8}  \end{array}
\right)
\otimes I_3
\end{equation}
The filtered state after the application of this 
filter is obtained as 
\begin{equation}
\rho^{f}(x,t)=(L_3\otimes M_3) \rho(x,t) (L_3\otimes
M_3)^{\dagger}
\label{filter1}
\end{equation}
We now apply the Choi map given in
equation~(\ref{eqn:1}) on the first system via
$\phi\otimes I_3 $ to the filtered as well as non-filtered 
density operator to obtain
\begin{eqnarray}
\rho(x,t) &\stackrel{\phi\otimes
I_3}{\longrightarrow}& \rho_{\rm Choi}(x,t)
\\
\rho^{f}(x,t) &\stackrel{\phi\otimes
I_3}{\longrightarrow} &
 \rho^{f}_{\rm Choi}(x,t)
\end{eqnarray}
For the operator $\rho^{f}_{\rm Choi}(x,t)$ for $t
=\frac{1}{20}$ and for $0.6044 < x < 0.6554$, the
minimum eigenvalue turns out to be negative,
indicating that the state is entangled.  On the
other hand, the minimum
eigen value of the operator $\rho_{\rm
Choi}(x,t)$ which is obtained by the
application of Choi map without filtering, is
positive. This shows that the local filter
defined in equation~(\ref{filter1}) has converted
the state $\rho(x,t)$ whose entanglement was not
detectable via the Choi map into a state
$\rho^{f}(x,t)$ whose entanglement is detectable
via the Choi map.

An important example of bound entangled states is
provided by the well known UPB construction known
as `TILES'~\cite{B2} for a $3\otimes 3$ system
\begin{eqnarray}\label{eq:tiles}
|\psi_0\rangle
&=&\frac{1}{\sqrt{2}}|0\rangle\left(|0\rangle-|1\rangle\right),\quad
|\psi_2\rangle =
\frac{1}{\sqrt{2}}|2\rangle\left(|1\rangle-|2\rangle\right),
\nonumber\\ |\psi_1\rangle &=&
\frac{1}{\sqrt{2}}\left(|0\rangle-|1\rangle\right)|2\rangle,\quad
|\psi_3\rangle =
\frac{1}{\sqrt{2}}\left(|1\rangle-|2\rangle\right)|0\rangle,\nonumber\\
|\psi_4\rangle &=&
\frac{1}{3}\left(|0\rangle+|1\rangle+|2\rangle\right)
\left(|0\rangle+|1\rangle+|2\rangle\right)
\end{eqnarray} The mixed state
\begin{equation}\label{eq:tilesstate}
\rho^{\rm upb}=\frac{1}{4}\left(I_9-\sum_{i=0}^4|
\psi_i\rangle\langle\psi_i|\right).
\end{equation} 
provides  an example of a PPT  entangled
state~\cite{B2}.  Choi maps, applied directly, can
not detect the entanglement of such states.
Consider the local filter
\begin{equation}
L_3^{\prime}\otimes M_3^{\prime} = 
I_3 \otimes  
\begin{pmatrix}
\frac{1}{\sqrt{2}} & 0 & \frac{1}{\sqrt{2}} \\
0 & 1 & 0 \\
-\frac{1}{\sqrt{2}} & 0 & \frac{1}{\sqrt{2}}
\end{pmatrix}, 
\end{equation}
Applying this filter gives a new filtered state
given by 
\begin{equation}
\rho^{f\,\rm upb} = (L_3^{\prime} \otimes M_3^{\prime}) \rho^{\rm upb}
(L_3^{\prime} \otimes M_3^{\prime})^{\dagger}
\end{equation}
We now apply the second Choi map given in
equation~(\ref{eqn:2}) on the second system via
$I_3\otimes\psi$ to the filtered as well as non-filtered density 
operator to obtain
\begin{eqnarray}
\rho^{\rm upb} &\stackrel{
I_3\otimes \psi}{\longrightarrow}& \rho^{\rm upb}_{\rm Choi}
\\
\rho^{f\,\rm upb} &\stackrel{
I_3\otimes \psi}{\longrightarrow} &
 \rho^{f~\rm upb}_{\rm Choi}.
\end{eqnarray}
The operator $\rho^{f~\rm upb}_{\rm Choi}$ has a
negative eigen value which reveals the
entanglement of $\rho^{\rm upb}$ while
$\rho^{\rm upb}_{\rm Choi}$ does not have a negative
eigen value. This shows that the entanglement of 
the state $\rho^{\rm upb}$ is not directly revealed
by the Choi map, however, it can be filtered into
a state that is detected by the Choi map. 
This is directly related to the construction given
in terms of the automorphisms
in~\cite{PhysRevA.87.012318} and is much simpler
than the  construction given in~\cite{T1}.
\section{Implementation of local filters}
\label{implementation}
\subsection{General scheme}
\label{implementation-general}
We now turn to the question of the physical
interpretation and implementation of the local quantum
filtration process introduced in the previous
section. A filter is defined through its action
given in equation~(\ref{filter}) and comprises of
invertible operators $L$ and $M$ where $L$ acts
locally on $\mathcal{H}_A$ (the Hilbert space of
Alice) and $M$ acts locally on
on $\mathcal{H}_B$ (the Hilbert space of Bob). We 
choose the standard bases
in $\mathcal{H}_A$ and $\mathcal{H}_B$.  Each of
these operators has a singular valued
decomposition given by $L=U_1 D_1 V_1$ and $M=U_2
D_2 V_2$. Here $U_1,U_2,V_1,V_2$ are unitary
operators and $D_1, D_2$ are diagonal with real
positive definite diagonal entries.  The  unitary
operators correspond to Hamiltonian evolutions and
can hence be physically realized in a
straightforward way.  We therefore focus here on
the implementation of diagonal matrices $D_1$ and
$D_2$.

Consider the implementation of $D_1$ on
$\mathcal{H}_A$. The diagonal matrix $D_1={\rm
Diag} [d_1,d_2,\cdots d_n]$ has diagonal entries
$d_j$ such that $0<d_j\leq 1, \, (j=1\cdots n)$.  We
now show that such a  $D_1$ can be implemented by
first extending the Hilbert space by adding one
qubit as ancilla and then measuring an appropriate
projection operator $P$.  To achieve this, we
first consider a set of $n$ orthogonal but
un-normalized vectors of the form
$\ket{u_j}=\sqrt{d_j}\ket{j}$ in the $n$
dimensional system  Hilbert space. We extend each
of these vectors into a $2n$ dimensional
Hilbert space to form a new set of
vectors $\{ \ket{\xi_j} \}$ given by
$\ket{\xi_j}=\sqrt{d_j}\ket{j}+
\sqrt{1-d_j}\ket{j+n}$.
In addition to being mutually orthogonal,
these vectors are also normalized. Thus we have
constructed an orthonormal set of $n$ vectors
in a $2n$ dimensional Hilbert space formed from
the $n$-dimensional system and a two-dimensional
ancilla.
Corresponding to each of these vectors, we 
can construct a projection operator $P_j$
given by
\begin{eqnarray}
P_j &=& \ket{\xi_j}\bra{\xi_j}\nonumber\\
    &=& \left(
\begin{array}{c|c}
\eta_j& \delta_j \\\hline
\delta_j & \eta_j^{\prime} 
\end{array} \right)_{2n \times 2n}
\label{topleft}
\end{eqnarray}
where the $n\times n$ matrices are given by:
\begin{eqnarray}
 \eta_j &=& d_j\ket{j}\bra{j}, 
\nonumber\\
 \eta_j^{\prime} &=& \left(1-d_j\right)\ket{j}\bra{j},
\nonumber\\
 \delta_j &=& \sqrt{d_j\left(1-d_j\right)}\ket{j}\bra{j} 
\end{eqnarray}
The projection operator obtained by adding these
mutually orthogonal projectors  is given
by
\begin{equation}
P = 
\sum_{j=1}^n P_j 
= \left(
\begin{array}{c|c}
D_1 & \Delta \\\hline
\Delta & D_1^{\prime} 
\end{array} \right)_{2n \times 2n}
\label{theprojector}
\end{equation}
where $D_1=\eta_1+\cdots+\eta_n$, is the
original operator that we wanted to implement,
$D_1^{\prime}=\eta_1^{\prime}+\cdots+\eta_n^{\prime}$ is a
complementary operator obtained from $D_1$ and
$\Delta=\delta_1+\cdots+\delta_n$ represents the
cross terms.

Now consider the system to be in an arbitrary
state $\rho_A$ and  the one-qubit  ancilla to be
in the state 
$\ket{0}\bra{0}$. Consider a measurement of $P$ on
this composite system. If the outcome of the
measurement is positive, we retain the state. The
state after such a selection is given by  the
action of the projection operator $P$ on the
composite state:
\begin{eqnarray}
 && P\left(\ket{0}\bra{0} \otimes \rho_A\right)
 P \nonumber\\ &=& \left(
\begin{array}{c|c}
D_1\rho_A {D_1} & D_1\rho_A {\Delta} \\\hline
{\Delta} \rho_A D_1 &  \Delta \rho_A \Delta 
\end{array} \right)
\label{map}
\end{eqnarray}
\begin{figure}[h]
\begin{center}
\includegraphics[]{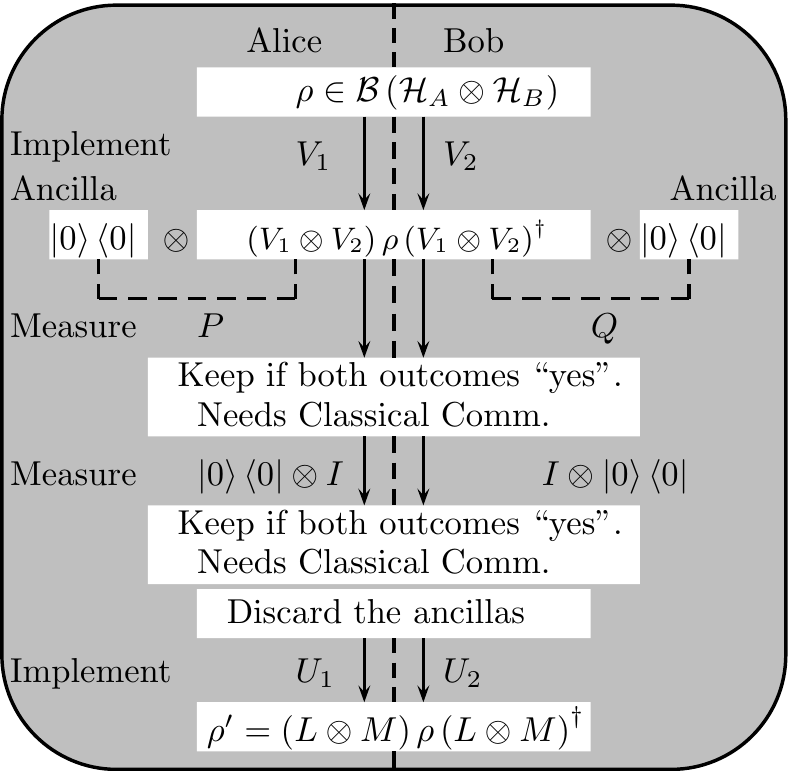}
\end{center}
\caption{Schematic diagram for performing the
local filtration via measurements.}
\label{filterpic}
\end{figure}
We further carry out  a projective measurement 
of the projector operator $\ket{0}\bra{0}$ on the ancilla qubit 
and retain the state only if the outcome  is positive.
We now discard the ancilla which is anyway decoupled from the
system. This completes the implementation of the map $D_1$ on the system
density operator $\rho_A$. Thus to realize the operator
$D_1$ in the space of the system,  we need to carry out two
projective measurements: first measure the projector $P$ in
the combined space of the system and ancilla and retain  the
state if the answer is positive and then make another
measurement of the projector $\vert 0 \rangle \langle 0
\vert$ only in the ancilla space and retain the state if the
outcome is positive.  As mentioned earlier, to implement $L$
we need to implement the unitaries $U_1$ and $V_1$ which can
be accomplished via the standard Hamiltonian evolution.  For
the implementation of the map $M$ on the second system in
the Hilbert space $\cal {H}_B$ we follow an analogous
procedure. The non-unitary part represented by $D_2$ is
implemented via two projective measurements after adding an
ancilla qubit and the unitaries $U_2$ and $V_2$ involved in
the singular decomposition of $M$ are implemented via a
Hamiltonian evolution.

In an actual implementation we require an ensemble
with several copies of the shared state $\rho$
between Alice and Bob. The protocol works as
follows:
\begin{itemize}
\item {\bf Step 1:} Alice implements the unitary $V_1$ on her
part of the state and Bob implements the unitary
$V_2$ on his part of the state. 
\item {\bf Step 2:} Alice and
Bob each attach a one-qubit ancilla prepared in a
state $\vert 0 \rangle$, to their part of the
state. Alice measures projector $P$ corresponding
to $D_1$ followed by a measurement of $ \vert 0
\rangle \langle 0 \vert $ on her ancilla qubit. Bob
measures the projector $Q$ corresponding to $D_2$
followed by a measurement of $ \vert 0 \rangle
\langle 0 \vert $ on his ancilla qubit. They retain
the state if all the four measurements give
positive results.  Otherwise they discard the
state.  
\item {\bf Step 3:} If they retain the state in the
previous step,
each one of them discards the ancilla qubits and
then Alice implements $U_1$ on her part of the
state and Bob implements $U_2$ on his part of the
state. They repeat this process on all the copies
of $\rho$ to obtain the new filtered ensemble.  
\end{itemize}
This protocol obviously requires classical
communication between Alice and Bob because they
need to know the outcome of the measurements that the
other performs. The situation is schematically
depicted in Figure~\ref{filterpic}.
\subsection{Filtration of two-qubit systems}
\label{implementation-two}
An interesting example of quantum filtration was
introduced by Gisin~\cite{MR1371104} for an
entangled  mixed state of two spin-$\frac{1}{2}$
particles not violating the Bell-CHSH inequality.
In this scheme by using a polarized beam splitter
one can convert such an input state to an output
state which remains entangled but, this time, its
entanglement can be detected by a Bell inequality
violation.  We recast this filtration scheme and
connect it with our results.

Let us suppose that Alice and Bob share a
$2\otimes 2$ system $\rho \in
\mathcal{B}\left(\mathcal{H}_A\otimes
\mathcal{H}_B\right)$ between themselves.
Interpreting the Gisin filter in our formalism
reveals that in that case $\rho \mapsto (L_2 \otimes
M_2) \rho (L_2 \otimes M_2)^\dagger$, where the
operators $L_2$ and $M_2$ are given  by 
\begin{equation}
L_2=\begin{pmatrix}
  \kappa && 0\\
   0 && 1
\end{pmatrix} \quad {\rm and} \quad
 M_2=\begin{pmatrix}
    1 && 0\\
    0 && \kappa
   \end{pmatrix}  
\end{equation}
In the notation of
reference~\cite{MR1371104},
$\kappa=\sqrt{\frac{\beta}{\alpha}}$, $\alpha$ and
$\beta$ are two real numbers such that
$\alpha>\beta >0$ and $\alpha^2 +\beta^2=1$ so
that $0< \kappa < 1$.  Let us consider the
implementation of the non-unitary operator $L_2$.
Since the operator $L_2$ is acting locally, we need
to look for the map 
\begin{equation*}
\rho_A \mapsto L_2\rho_A L_2.
\end{equation*}
Since $L_2$ is a diagonal matrix ($L_2^{\dagger}=L_2$), we do not need to
undertake a singular value decomposition.  We
directly  introduce two mutually orthogonal but
un-normalized vectors
\begin{equation}
 \ket{u_1}=\begin{pmatrix}
              \sqrt{\kappa}\\ 0
             \end{pmatrix}, \quad
 \ket{u_2}=\begin{pmatrix}
              0 \\ 1
             \end{pmatrix}.
\end{equation}
As discussed in 
Section~\ref{implementation-general}, by adding
a one-qubit ancilla these  two-dimensional vectors
can be extended into four-dimensional orthonormal
vectors
\begin{equation}
 \ket{\xi_1} = \begin{pmatrix}
              \sqrt{\kappa}\\ 0 \\
              \sqrt{1-\kappa} \\ 0
             \end{pmatrix}, \quad
 \ket{\xi_2} = \begin{pmatrix}
              0 \\ 1 \\ 0 \\ 0
             \end{pmatrix}		
\end{equation}
Constructing the corresponding mutually orthogonal
projectors results in 
\begin{eqnarray}
 P_1 &=& \begin{pmatrix}
          \kappa & 0 &
          \sqrt{\kappa\left(1-\kappa\right)} & 0 \\
          0 & 0 & 0 & 0\\
         \sqrt{\kappa\left(1-\kappa\right)} &
          0 & 1-\kappa & 0\\
          0 & 0 & 0 & 0
         \end{pmatrix}\nonumber\\
  P_2 &=& \begin{pmatrix}
           0 & 0 & 0 & 0\\
           0 & 1 & 0 & 0\\
           0 & 0 & 0 & 0\\
           0 & 0 & 0 & 0
          \end{pmatrix}       
\end{eqnarray}
The required projector $P = P_1+P_2$ is then
given by
\begin{eqnarray}
P = \left(
\begin{array}{c|c}
L_2 & \Delta_2\\\hline 
\Delta_2& L_2^{\prime} \end{array} \right)
\end{eqnarray}
Here $L_2^{\prime}$ and 
$\Delta_2$ are defined following the definitions
below Equation~(\ref{theprojector}).

Consider now the action of this projector on a
state where the system is in an arbitrary state
$\rho_A$ and the ancilla qubit is in the state 
$\rho_a=\ket{0}\bra{0}$.
\begin{equation}
P (\ket{0}\bra{0}\otimes \rho_A)P=
\left(  
\begin{array}{c|c}
L_2\rho_A L_2 & L_2 \rho_A {\Delta_2} \\\hline
{\Delta_2}\rho_A L_2 &  \Delta_2 \rho_A \Delta_2
\end{array} \right)\label{ar}
\end{equation}
This is the result of measurement of $P$ on the
composite system for the case when the outcome is
positive. To extract the top left block of the
above matrix, we perform another projective
measurement in the ancilla space and retain it if
the outcome of measurement of $\vert 0 \rangle
\langle 0 \vert$ is positive.  Bob does a similar
exercise to implement $M_2$ in his laboratory and both
Alice and Bob retain the state only when the
outcomes of measurement of all four projectors are
positive. This completes the protocol. 
\subsection{The case of $3\otimes 3$ systems}
\label{implementation-three}
In Section~\ref{filtration-ent} we demonstrated
the role of local filters in strengthening the
entanglement detection capabilities of
entanglement witnesses for $3\otimes 3$ systems.
Here we delineate the implementation of such a
filtration process for these systems.  We begin by
discussing the filtration on a single three-level
system $\rho_A \in
\mathcal{B}\left(\mathcal{H}_A\right)$.  We
consider the following transformation
$\rho_A\mapsto L_3\rho_A L_3^\dag$ where $L=U_1D_1V_1$
is the singular value decomposition of $L$ and
\begin{equation}
 D_1=\begin{pmatrix}
d_1 &&\\
& d_2&\\
&& d_3
\end{pmatrix}
\end{equation}
with $0< d_1,d_2,d_3 \leq 1 $. 

We now introduce three
mutually orthogonal and un-normalized vectors,
in a three-dimensional Hilbert space
\begin{equation}
\ket{u_1} = \begin{pmatrix}
                \sqrt{d_1} \\ 0 \\ 0 
               \end{pmatrix},
 \ket{u_2} = \begin{pmatrix}
                0\\ \sqrt{d_2} \\ 0 
               \end{pmatrix},
 \ket{u_3} = \begin{pmatrix}
                0 \\ 0 \\ \sqrt{d_3} 
               \end{pmatrix}
\end{equation}
Now introducing a qubit as an ancilla  and extending the above
vectors in the composite Hilbert space of 6 dimensions, we obtain
\begin{equation}
 \ket{\xi_1} = \begin{pmatrix}
                \sqrt{d_1} \\ 0 \\ 0 \\
                \!\!
\sqrt{1-d_1} \\ 0 \\ 0
               \end{pmatrix},\,
 \ket{\xi_2} = \begin{pmatrix}
                0\\ \sqrt{d_2} \\ 0 \\
                0 \\\!\!\sqrt{1-d_2} \\ 0
               \end{pmatrix}, \,
 \ket{\xi_3} = \begin{pmatrix}
                0 \\ 0 \\ \sqrt{d_3} \\
                0 \\ 0 \\ \!\sqrt{1-d_3}
               \end{pmatrix}            
\end{equation}
which are orthonormal.

The projection operators corresponding to these
vectors are given by  
$P_j=\ket{\xi_j}\bra{\xi_j}$  and can be written
explicitly as 
\begin{equation*}
 P_1 = \begin{pmatrix}
          d_1 & 0 & 0 & \sqrt{d_1 \left(1-d_1\right)} & 0 & 0\\
          0 & 0 & 0 & 0 & 0 & 0\\
          0 & 0 & 0 & 0 & 0 & 0\\
          \sqrt{d_1 \left(1-d_1\right)} & 0 & 0 & 1-d_1 & 0 & 0\\
          0 & 0 & 0 & 0 & 0 & 0\\
          0 & 0 & 0 & 0 & 0 & 0
         \end{pmatrix}
\end{equation*}
\begin{equation*}
 P_2 = \begin{pmatrix}
          0 & 0 & 0 & 0 & 0 & 0\\
          0 & d_2 & 0 & 0 & \sqrt{d_2\left(1-d_2\right)} & 0\\
          0 & 0 & 0 & 0 & 0 & 0\\
          0 & 0 & 0 & 0 & 0 & 0\\
          0 & \sqrt{d_2\left(1-d_2\right)} & 0 & 0 & 1-d_2 & 0\\
          0 & 0 & 0 & 0 & 0 & 0
         \end{pmatrix}
\end{equation*}
\begin{equation}
 P_3 = \begin{pmatrix}
          0 & 0 & 0 & 0 & 0 & 0\\
          0 & 0 & 0 & 0 & 0 & 0\\
          0 & 0 & d_3 & 0 & 0 & \sqrt{d_3\left(1-d_3\right)}\\
          0 & 0 & 0 & 0 & 0 & 0\\
          0 & 0 & 0 & 0 & 0 & 0\\
          0 & 0 & \sqrt{d_3\left(1-d_3\right)} & 0 & 0 & 1-d_3
         \end{pmatrix}        
\end{equation}
The projector $P$ corresponding to the operator $L_3$ is thus given by
\begin{equation}
P = P_1+P_2+P_3
  = \left(
\begin{array}{c|c}
L_3 & \Delta_3\\\hline 
\Delta_3& L_3^{\prime} \end{array} \right)
\end{equation}
For any operator
$\rho_A\in\mathcal{B}(\mathbb{C}^3)$ we use a
one-qubit ancillary  system in the state
$\ket{0}\bra{0}$ and  measure the projector $P$.
This leads us to an equation which is the same as
the Equation~(\ref{ar}) with $L_2$ and $\Delta_2$
replaced by $L_3$ and $\Delta_3$.  Similarly we
can do the analysis from Bob's point of view and
arrive at the projector $Q$ corresponding to $D_2$
($M_3=U_2D_2V_2$). Using this projector and a one
qubit ancilla he sets up his measurements.  Then
they both follow the protocol steps 1-3 given in
the last part of
Section~\ref{implementation-general} to complete
the filtration process and obtain the new joint
density operator.
\section{Conclusions}
\label{conc}
In this work, we discussed the role of local
filters in transforming one PPT entangled state to
another PPT entangled state. It may turn out that the
entanglement of the new state is detectable by a P
map which is not CP while the entanglement of the
original state is not detectable by the map.  It
is in this sense that local filters can enhance
the power of an entanglement witness in detecting
entanglement. We give two concrete examples where
this actually occurs. In  the first example, a new
class of bound entangled states becomes detectable
by the Choi map and in the second example PPT
entangled states in the orthogonal complement of
UPB become detectable by the Choi map.

We then undertook the analysis of these
filtration schemes as explicit local projective
measurements coupled with local unitaries. It
turns out that we need to add a one-qubit ancilla 
for both the parties involved in order to
implement the non-unitary part of these filters as 
local measurements. We have constructed explicit 
projection operators corresponding to the filters
that we have used.

A point worth emphasizing is that these local
filters do not change the NPT or PPT status of a
state. Gisin exploited this fact to convert NPT
states of two qubits which do not violate Bell's
inequalities into the ones which violate Bell's
inequalities. We have used these filters to
convert one PPT entangled state into another PPT
entangled state such that the PPT entanglement is
detectable by a given entanglement witness.
\begin{acknowledgements}
DD acknowledges CSIR India and RS
acknowledges NBHM India for financial support.
\end{acknowledgements}
\def\Dbar{\leavevmode\lower.6ex\hbox to 0pt{\hskip-.23ex \accent"16\hss}D}
  \def\Dbar{\leavevmode\lower.6ex\hbox to 0pt{\hskip-.23ex \accent"16\hss}D}
  \def\polhk#1{\setbox0=\hbox{#1}{\ooalign{\hidewidth
  \lower1.5ex\hbox{`}\hidewidth\crcr\unhbox0}}} \def\cprime{$'$}
\end{document}